\documentclass[preprint]{elsarticle}

\usepackage{lineno,hyperref}
\usepackage{graphicx,amsmath,amssymb,amsthm,bm}
\modulolinenumbers[5]

\journal{}

\newtheorem{definition}{Definition}
\newtheorem{proposition}{Proposition}
\newtheorem{lemma}{Lemma}
\newtheorem{theorem}{Theorem}
\newtheorem{corollary}{Corollary}

\begin{document}

\begin{frontmatter}

\title{Zero-determinant strategies in repeated continuously-relaxed games}

\author{Masahiko Ueda\corref{mycorrespondingauthor}}
\ead{m.ueda@yamaguchi-u.ac.jp}
\address{Graduate School of Sciences and Technology for Innovation, Yamaguchi University, Yamaguchi 753-8511, Japan}

\author{Ayaka Fujita}
\address{Graduate School of Sciences and Technology for Innovation, Yamaguchi University, Yamaguchi 753-8511, Japan}

\begin{abstract}
Mixed extension has played an important role in game theory, especially in the proof of the existence of Nash equilibria in strategic form games.
Mixed extension can be regarded as continuous relaxation of a strategic form game.
Recently, in repeated games, a class of behavior strategies, called zero-determinant strategies, was introduced.
Zero-determinant strategies control payoffs of players by unilaterally enforcing linear relations between payoffs.
There are many attempts to extend zero-determinant strategies so as to apply them to broader situations.
Here, we extend zero-determinant strategies to repeated games where action sets of players in stage game are continuously relaxed.
We see that continuous relaxation broadens the range of possible zero-determinant strategies, compared to the original repeated games.
Furthermore, we introduce a special type of zero-determinant strategies, called one-point zero-determinant strategies, which repeat only one continuously-relaxed action in all rounds.
By investigating several examples, we show that some property of mixed-strategy Nash equilibria can be reinterpreted as a payoff-control property of one-point zero-determinant strategies.
\end{abstract}


\end{frontmatter}


\section{Introduction}
\label{sec:intro}
In game theory, mixed strategies, which play pure strategies probabilistically, have played an important role \cite{FudTir1991,OsbRub1994}.
Nash proved the existence of a Nash equilibrium in mixed extensions of strategic form games \cite{Nas1950}.
This result is true even if a corresponding pure strategy game does not contain any Nash equilibria.
In two-player zero-sum games, the existence of Nash equilibria is known as a minimax theorem, and it is a special case of duality theorem in linear programming \cite{Adl2013}.
Linear programming is a continuous relaxation of integer linear programming.
Continuous relaxation is a technique in optimization theory which relaxes discrete variables to continuous variables, and generally makes problems easier \cite{CLRS2022}.
Mixed extension can be regarded as continuous relaxation of a strategic form game.
For strategic form games with finite action spaces, mixed extension relaxes discrete actions to continuous variables interpreted as probability.

In repeated games, a class of behavior strategies, called zero-determinant (ZD) strategies, was discovered in 2012 \cite{PreDys2012}.
A ZD strategy controls payoffs in a repeated game by unilaterally enforcing linear relationships between payoffs \cite{HRZ2015,HTS2015,HDND2016,MamIch2020,Ued2022c}.
Although ZD strategies were originally introduced in the repeated prisoner's dilemma game, the concept of ZD strategies was later extended to arbitrary stage games \cite{McAHau2016,McAHau2017}.
Concurrently, a necessary and sufficient condition for the existence of ZD strategies was specified for stage games where each action set is a finite set \cite{Ued2022b}.
Whereas this condition cannot be applied to stage games where action sets of some players are infinite sets, McAvoy and Hauert also provided a procedure to explicitly construct ZD strategies (called two-point ZD strategies) for arbitrary stage games when some sufficient condition is satisfied \cite{McAHau2016}.
Interestingly, the latter sufficient condition coincides with the former necessary and sufficient condition \cite{Ued2025}.
Therefore, if action sets of all players are finite sets, the existence of two-point ZD strategies is equivalent to the existence of ZD strategies.

In this paper, we introduce ZD strategies for repeated versions of mixed extensions of strategic form games, when the original action set of each player is a finite set.
Such extension can be regarded as ZD strategies on continuously-relaxed action space.
By using the technique proposed by Ref. \cite{McAHau2016}, we can construct two-point ZD strategies in our setting.
We then investigate difference between an existence condition of such two-point ZD strategies on continuously-relaxed action sets and that of ZD strategies on the original (finite) action sets.
Through construction of an example, we see that continuous relaxation broadens the range of possible ZD strategies, compared to the original repeated games.
Furthermore, as a special case of two-point ZD strategies, we introduce the concept of one-point ZD strategies, which use only one continuously-relaxed action.
We show that, in several examples, a property of mixed-strategy Nash equilibria can be interpreted as a payoff-control property of one-point ZD strategies.

This paper is organized as follows.
In Section \ref{sec:model}, we introduce our model of repeated games with continuous relaxation.
In Section \ref{sec:preliminaries}, we introduce ZD strategies, and previous results on the existence of two-point ZD strategies.
In Section \ref{sec:existence}, we investigate difference between an existence condition of two-point ZD strategies on continuously-relaxed action sets and an existence condition of ZD strategies on the original finite action sets.
Several examples of two-point ZD strategies are provided in this section.
In Section \ref{sec:one-point}, we introduce the concept of one-point ZD strategies.
By analyzing several examples, we discuss a relation between a property of mixed-strategy Nash equilibria and a property of one-point ZD strategies in the section.
Section \ref{sec:conclusion} is devoted to concluding remarks.

\section{Model}
\label{sec:model}
We first introduce a strategic game.
We write the set of players as $\mathcal{N}$.
We also write the set of actions of player $j\in \mathcal{N}$ as $A_j$.
The payoff function of player $j\in \mathcal{N}$ is written as $s_j: \prod_{k\in \mathcal{N}}A_k \rightarrow \mathbb{R}$.
A strategic game is defined by these three elements as $G:=\left( \mathcal{N}, \left\{ A_j \right\}_{j\in \mathcal{N}}, \left\{ s_j \right\}_{j\in \mathcal{N}} \right)$ \cite{FudTir1991,OsbRub1994}.
In this paper, we assume that $\mathcal{N}$ and $A_j$ $\left( \forall j \in \mathcal{N} \right)$ are finite sets.
We introduce notations $\mathcal{A}:= \prod_{k\in \mathcal{N}}A_k$ and $A_{-j}:= \prod_{k\neq j}A_k$.
We write an action profile as $\bm{a}:=\left( a_k \right)_{k\in \mathcal{N}} \in \mathcal{A}$.
For simplicity, we also introduce the notation $a_{-j}:= \left( a_k \right)_{k\neq j} \in A_{-j}$.
When we focus on an action of player $j$ in $\bm{a}$, we write $\bm{a}=(a_j, a_{-j})$.

Next, we introduce mixed extensions of the strategic games.
A mixed extension of the strategic game $G$ is defined as $\tilde{G}:=\left( \mathcal{N}, \left\{ \Delta\left(A_j\right) \right\}_{j\in \mathcal{N}}, \left\{ u_j \right\}_{j\in \mathcal{N}} \right)$, where $\Delta\left(A_j\right)$ is the set of all probability distributions on $A_j$, and $u_j$ is the expected payoff of player $j\in \mathcal{N}$.
When we introduce the notations $p_j:=\left( p_j\left( a_j \right) \right)_{a_j \in A_j} \in \Delta\left(A_j\right)$ and $\bm{p}:=\left( p_k \right)_{k\in \mathcal{N}} \in \prod_{k \in \mathcal{N}} \Delta\left(A_k\right)$, we can write
\begin{align}
 u_j \left( \bm{p} \right) &= \sum_{\bm{a}} \left\{ \prod_{k\in \mathcal{N}} p_k\left( a_k \right) \right\} s_j \left( \bm{a} \right).
\end{align}
Similarly as above, we introduce the notations $\tilde{\Delta}\left( \mathcal{A} \right) := \prod_{k \in \mathcal{N}} \Delta\left(A_k\right)$, $\tilde{\Delta}\left( A_{-j} \right) := \prod_{k \neq j} \Delta\left(A_k\right)$, and $p_{-j}:= \left( p_k \right)_{k\neq j} \in \tilde{\Delta}\left( A_{-j} \right)$.
The strategy $p_j\in \Delta\left(A_j\right)$ is called a \emph{mixed strategy} of player $j\in \mathcal{N}$.
Mixed strategies can be regarded as continuous relaxation of actions in the original game $G$.

There are many interpretation on mixed strategies \cite{OsbRub1994}.
In this paper, we interpret a mixed strategy $p_j$ as a strategy realized by population $j$.
That is, a player $j\in \mathcal{N}$ is a population of individuals, and games are played by populations.
Therefore, a mixed strategy of each player can be explicitly observed.

We then consider the repeated version of mixed extensions of the strategic games.
We write a mixed-strategy profile at $t$-th round as $\bm{p}^{(t)}$.
When we write a history of mixed-strategy profiles between $t^\prime$-th round and $t$-th round with $t^\prime \leq t$ as $h_{[t^\prime,t]}:= \left\{ \bm{p}^{(s)} \right\}_{s=t^\prime}^t$, a behavior strategy of player $j$ in the repeated game is defined as
\begin{align}
 \mathcal{T}_j &:= \left\{ T^{(t)}_j \left( p^{(t)}_j | h_{[1,t-1]} \right) | t\in \mathbb{N}, p^{(t)}_j \in \Delta\left(A_j\right), h_{[1,t-1]} \in \tilde{\Delta}\left( \mathcal{A} \right)^{t-1} \right\},
 \label{eq:behavior_strategy}
\end{align}
where $T^{(t)}_j \left( p^{(t)}_j | h_{[1,t-1]} \right)$ is a conditional probability density function of player $j$ at $t$-th round using a mixed strategy $p^{(t)}_j$ when a history of mixed-strategy profiles is $h_{[1,t-1]}$.
We write the expected value of the quantity $Q$ with respect to the behavior strategies $\left\{ \mathcal{T}_k \right\}_{k\in \mathcal{N}}$ of all players as $\mathbb{E} \left[ Q \right]$.
The payoffs in the repeated game is defined as
\begin{align}
 \mathcal{U}_j &:= \lim_{T\rightarrow \infty} \frac{1}{T} \sum_{t=1}^T \mathbb{E} \left[ u_j \left( \bm{p}^{(t)} \right) \right] \quad (j \in  \mathcal{N}).
\end{align}
That is, in this paper we assume that there is no discounting.

\section{Preliminaries}
\label{sec:preliminaries}
In this section, we introduce the concept of zero-determinant strategies and previous studies about the way to specify the existence of zero-determinant strategies.
Below we write the Kronecker delta and the Dirac delta function as $\delta_{a, a^\prime}$ and $\delta(p)$, respectively.

A time-independent memory-one strategy of player $j$ is defined as $\mathcal{T}_j$ with
\begin{align}
T^{(t)}_j \left( p^{(t)}_j | h_{[1,t-1]} \right) &= T_j \left( p^{(t)}_j | h_{[t-1,t-1]} \right)
\end{align}
for $\forall t\in \mathbb{N}, \forall p^{(t)}_j \in \Delta\left(A_j\right), \forall h_{[1,t-1]} \in \tilde{\Delta}\left( \mathcal{A} \right)^{t-1}$.
For such time-independent memory-one strategies, we introduce zero-determinant strategies, by arranging the definition in arbitrary action space \cite{McAHau2016} for our mixed strategy space.
\begin{definition}
\label{def:ZD}
A time-independent memory-one strategy of player $j$ is a \emph{zero-determinant (ZD) strategy} when $T_j$ satisfies
\begin{align}
 \int \psi(p_j) T_j\left( p_j | \bm{p}^\prime \right) dp_j - \psi\left( p^\prime_j \right) &= \sum_{k\in \mathcal{N}} \alpha_{k} u_{k} \left( \bm{p}^{\prime} \right) + \alpha_0
 \label{eq:AS}
\end{align}
with some coefficients $\left\{ \alpha_{b} \right\}$ and some bounded function $\psi(\cdot)$.
\end{definition}
When action sets are not finite sets as in our case $\Delta\left(A_k\right)$, ZD strategies are also called as \emph{autocratic strategies}.
It has been known that a ZD strategy (\ref{eq:AS}) unilaterally enforces a linear relation between payoffs \cite{McAHau2016}
\begin{align}
 0 &= \sum_{k\in \mathcal{N}} \alpha_{k} \mathcal{U}_j + \alpha_0.
 \label{eq:linear_AS}
\end{align}
When a ZD strategy unilaterally enforces Eq. (\ref{eq:linear_AS}), we call it a ZD strategy controlling $\sum_{k \in \mathcal{N}} \alpha_{k} u_{k} + \alpha_0$.
The quantity such as $T_j\left( p_j | \bm{p}^\prime \right) - \delta(p_j-p^\prime_j)$ in the left-hand side of Eq. (\ref{eq:AS}) has been called as the Press-Dyson functions \cite{Aki2016,McAHau2016}, which describes difference between the time-independent memory-one strategy and the Repeat strategy.
Definition \ref{def:ZD} means that a linear combination of the Press-Dyson functions of a ZD strategy is described by a linear combination of payoff functions and a constant function.

Below we write $B\left( \bm{a} \right):=\sum_{k \in \mathcal{N}} \alpha_{k} s_{k} \left( \bm{a} \right) + \alpha_0$ and $\tilde{B}\left( \bm{p} \right):=\sum_{k \in \mathcal{N}} \alpha_{k} u_{k} \left( \bm{p} \right) + \alpha_0$.
It should be noted that
\begin{align}
 \tilde{B}\left( \bm{p} \right) &= \sum_{\bm{a}} \left\{ \prod_{k\in \mathcal{N}} p_k\left( a_k \right) \right\} B \left( \bm{a} \right).
 \label{eq:Btilde}
\end{align}
McAvoy and Hauert provided a sufficient condition for the existence of ZD strategies in arbitrary action space \cite{McAHau2016}.
Here we again arrange their result in order to apply it to our mixed strategy space.
\begin{proposition}
\label{prop:two-point}
If there exist two mixed strategies $\underline{p}_j, \overline{p}_j \in \Delta(A_j)$ and a constant $W>0$ such that
\begin{align}
 -W &\leq \tilde{B}\left( \underline{p}_j, p_{-j} \right) \leq 0 \quad (\forall p_{-j} \in \tilde{\Delta}(A_{-j})) \nonumber \\
 0 &\leq \tilde{B}\left( \overline{p}_j, p_{-j} \right) \leq W \quad (\forall p_{-j} \in \tilde{\Delta}(A_{-j})),
 \label{eq:existence_condition}
\end{align}
then the memory-one strategy of player $j$,
\begin{align}
 T_j\left( p_j | \bm{p}^\prime \right) &= \delta\left( p_j-p_j^\prime \right) + \frac{1}{W} \tilde{B}\left( \bm{p}^{\prime} \right) \delta\left( p_j-\underline{p}_j \right) - \frac{1}{W} \tilde{B}\left( \bm{p}^{\prime} \right) \delta\left( p_j-\overline{p}_j \right) \quad (\forall p_j^\prime \in \Delta\left(A_j\right)^\prime, \forall p_{-j}^\prime \in \tilde{\Delta}(A_{-j})),
 \label{eq:transition_two-point}
\end{align}
where $\Delta\left(A_j\right)^\prime:=\left\{ \underline{p}_j, \overline{p}_j \right\}$ is a restricted action set of player $j$, is a ZD strategy controlling $\tilde{B}$.
\end{proposition}
We can easily check that the strategy (\ref{eq:transition_two-point}) satisfies Definition \ref{def:ZD} as
\begin{align}
 \int \psi(p_j) T_j\left( p_j | \bm{p}^\prime \right) dp_j - \psi\left( p^\prime_j \right) &= \left[ \psi\left( \underline{p}_j \right) - \psi\left( \overline{p}_j \right) \right] \frac{1}{W} \tilde{B}\left( \bm{p}^{\prime} \right) \quad (\forall p_j^\prime \in \Delta\left(A_j\right)^\prime, \forall p_{-j}^\prime \in \tilde{\Delta}(A_{-j}))
\end{align}
with arbitrary $\psi$ with $\psi\left( \underline{p}_j \right) \neq \psi\left( \overline{p}_j \right)$.
Such ZD strategies can be called as \emph{two-point ZD strategies}, because they use only two $p_j$.
It should be noted that the existence of $W$ is trivial for our case because the payoffs are bounded.

The condition (\ref{eq:existence_condition}) is also known to be a necessary condition for the existence of two-point ZD strategies \cite{Ued2022b,Ued2025}.
Indeed, if a two-point ZD strategy of player $j$,
\begin{align}
 T_j\left( p_j | \bm{p}^\prime \right) &= D\left( \bm{p}^{\prime} \right) \delta\left( p_j-p_j^{(1)} \right) + \left( 1-D\left( \bm{p}^{\prime} \right) \right) \delta\left( p_j-p_j^{(2)} \right) \quad \left( \forall p_j^\prime \in \left\{ p_j^{(1)}, p_j^{(2)} \right\}, \forall p_{-j}^\prime \in \tilde{\Delta}(A_{-j}) \right),
\end{align}
satisfying $0\leq D(\cdot ) \leq 1$ and
\begin{align}
 \int \psi(p_j) T_j\left( p_j | \bm{p}^\prime \right) dp_j - \psi\left( p^\prime_j \right) &= \tilde{B}\left( \bm{p}^{\prime} \right) \quad \left( \forall p_j^\prime \in \left\{ p_j^{(1)}, p_j^{(2)} \right\}, \forall p_{-j}^\prime \in \tilde{\Delta}(A_{-j}) \right)
\end{align}
exists, we obtain
\begin{align}
\tilde{B}\left( \bm{p}^{\prime} \right) &=  \psi\left( p_j^{(1)} \right) D\left( \bm{p}^{\prime} \right) + \psi\left( p_j^{(2)} \right) \left( 1-D\left( \bm{p}^{\prime} \right) \right) - \psi\left( p^\prime_j \right).
\end{align}
Especially, equalities
\begin{align}
\tilde{B}\left( p_j^{(1)}, p_{-j}^\prime  \right) &=  \left[ \psi\left( p_j^{(2)} \right) - \psi\left( p_j^{(1)} \right) \right] \left( 1-D\left( p_j^{(1)}, p_{-j}^\prime \right) \right) \nonumber \\
\tilde{B}\left( p_j^{(2)}, p_{-j}^\prime  \right) &=  - \left[ \psi\left( p_j^{(2)} \right) - \psi\left( p_j^{(1)} \right) \right] D\left( p_j^{(2)}, p_{-j}^\prime \right)
\end{align}
hold.
If $\psi\left( p_j^{(2)} \right) - \psi\left( p_j^{(1)} \right) \geq 0$, we obtain
\begin{align}
\tilde{B}\left( p_j^{(1)}, p_{-j}^\prime  \right) &\geq 0 \quad \left( \forall p_{-j}^\prime \in \tilde{\Delta}(A_{-j}) \right) \nonumber \\
\tilde{B}\left( p_j^{(2)}, p_{-j}^\prime  \right) &\leq 0  \quad \left( \forall p_{-j}^\prime \in \tilde{\Delta}(A_{-j}) \right).
\end{align}
If $\psi\left( p_j^{(2)} \right) - \psi\left( p_j^{(1)} \right) \leq 0$, we obtain
\begin{align}
\tilde{B}\left( p_j^{(1)}, p_{-j}^\prime  \right) &\leq 0 \quad \left( \forall p_{-j}^\prime \in \tilde{\Delta}(A_{-j}) \right) \nonumber \\
\tilde{B}\left( p_j^{(2)}, p_{-j}^\prime  \right) &\geq 0  \quad \left( \forall p_{-j}^\prime \in \tilde{\Delta}(A_{-j}) \right).
\end{align}
This result is summarized to the following proposition.
\begin{proposition}
\label{prop:two-point_necessity}
If a two-point ZD strategy of player $j$ controlling $\tilde{B}$ exists, then there exist two mixed strategies $\underline{p}_j, \overline{p}_j \in \Delta(A_j)$ such that
\begin{align}
 \tilde{B}\left( \underline{p}_j, p_{-j} \right) &\leq 0 \quad (\forall p_{-j} \in \tilde{\Delta}(A_{-j})) \nonumber \\
 \tilde{B}\left( \overline{p}_j, p_{-j} \right) &\geq 0 \quad (\forall p_{-j} \in \tilde{\Delta}(A_{-j})).
\end{align}
\end{proposition}
Therefore, we find that the condition (\ref{eq:existence_condition}) is a necessary and sufficient condition for the existence of two-point ZD strategies.

\section{Existence of two-point ZD strategies}
\label{sec:existence}
We first provide a sufficient condition for the existence of two-point ZD strategies in our mixed extension games.
\begin{theorem}
\label{thm:pureZDS}
If there exist two actions $\underline{a}_j, \overline{a}_j \in A_j$ such that
\begin{align}
 B\left( \underline{a}_j, a_{-j} \right) &\leq 0 \quad (\forall a_{-j} \in A_{-j}) \nonumber \\
 B\left( \overline{a}_j, a_{-j} \right) &\geq 0 \quad (\forall a_{-j} \in A_{-j}),
 \label{eq:existence_condition_pure}
\end{align}
then a two-point ZD strategy of player $j$ controlling $\tilde{B}$ exists.
\end{theorem}
\begin{proof}
When we consider
\begin{align}
 \underline{p}_j(a_j) &= \delta_{a_j, \underline{a}_j} \\
 \overline{p}_j(a_j) &= \delta_{a_j, \overline{a}_j},
\end{align}
we find
\begin{align}
 \tilde{B}\left( \underline{p}_j, p_{-j} \right) &= \sum_{a_{-j}} \left\{ \prod_{k\neq j} p_k\left( a_k \right) \right\} B \left(  \underline{a}_j, a_{-j} \right) \leq 0 \quad (\forall p_{-j} \in \tilde{\Delta}(A_{-j})) \\
 \tilde{B}\left( \overline{p}_j, p_{-j} \right) &= \sum_{a_{-j}} \left\{ \prod_{k\neq j} p_k\left( a_k \right) \right\} B \left(  \overline{a}_j, a_{-j} \right) \geq 0 \quad (\forall p_{-j} \in \tilde{\Delta}(A_{-j})).
\end{align}
Therefore, according to Proposition \ref{prop:two-point}, a two-point ZD strategy exists.
\end{proof}
We remark that the condition (\ref{eq:existence_condition_pure}) is known as a necessary and sufficient condition for the existence of ZD strategies in original strategic games (with finite action sets) \cite{Ued2022b}, and, moreover, it is equivalent to a necessary and sufficient condition for the existence of two-point ZD strategies in original strategic games \cite{Ued2025}.
Therefore, even if we consider mixed extension (continuous relaxation) of strategic games, the same type of payoff control as one in original strategic games is possible.

Next, we rewrite the condition (\ref{eq:existence_condition}) into a simpler form.
\begin{lemma}
\label{lem:existence_equivalent}
The existence of two mixed strategies $\underline{p}_j, \overline{p}_j \in \Delta(A_j)$ such that
\begin{align}
 \tilde{B}\left( \underline{p}_j, p_{-j} \right) &\leq 0 \quad (\forall p_{-j} \in \tilde{\Delta}(A_{-j})) \nonumber \\
 \tilde{B}\left( \overline{p}_j, p_{-j} \right) &\geq 0 \quad (\forall p_{-j} \in \tilde{\Delta}(A_{-j}))
 \label{eq:existence_condition_mod}
\end{align}
is equivalent to the existence of two mixed strategies $\underline{p}_j, \overline{p}_j \in \Delta(A_j)$ such that
\begin{align}
 \sum_{a_j}  \underline{p}_j(a_j) B\left( a_j, a_{-j} \right) &\leq 0 \quad (\forall a_{-j} \in A_{-j}) \nonumber \\
 \sum_{a_j}  \overline{p}_j(a_j) B\left( a_j, a_{-j} \right) &\geq 0 \quad (\forall a_{-j} \in A_{-j}).
 \label{eq:existence_condition_equivalent}
\end{align}
Furthermore, this condition is also equivalent to
\begin{align}
 \min_{p_j} \max_{a_{-j}} \sum_{a_j}  p_j(a_j) B\left( a_j, a_{-j} \right) &\leq 0 \nonumber \\
 \max_{p_j} \min_{a_{-j}} \sum_{a_j}  p_j(a_j) B\left( a_j, a_{-j} \right) &\geq 0.
 \label{eq:existence_condition_minmax}
\end{align}
\end{lemma}
\begin{proof}
If Eq. (\ref{eq:existence_condition_equivalent}) holds for some $\underline{p}_j, \overline{p}_j \in \Delta(A_j)$, they trivially satisfy Eq. (\ref{eq:existence_condition_mod}).

Suppose that Eq. (\ref{eq:existence_condition_mod}) holds for some $\underline{p}_j, \overline{p}_j \in \Delta(A_j)$.
If the inequality
\begin{align}
 \sum_{a_j}  \underline{p}_j(a_j) B\left( a_j, a_{-j}^\prime \right) &> 0
\end{align}
holds for some $a_{-j}^\prime \in A_{-j}$, when we consider
\begin{align}
 p_{-j}^\prime (a_{-j}) &:= \delta_{a_{-j}, a_{-j}^\prime} \in \tilde{\Delta}(A_{-j}),
\end{align}
we obtain
\begin{align}
 \tilde{B}\left( \underline{p}_j, p_{-j}^\prime \right) &>0,
\end{align}
leading to contradiction.
Therefore, we find
\begin{align}
 \sum_{a_j}  \underline{p}_j(a_j) B\left( a_j, a_{-j} \right) &\leq 0 \quad (\forall a_{-j} \in A_{-j}).
\end{align}
The inequality
\begin{align}
 \sum_{a_j}  \overline{p}_j(a_j) B\left( a_j, a_{-j} \right) &\geq 0 \quad (\forall a_{-j} \in A_{-j})
\end{align}
holds for a similar reason.

Furthermore, we remark that Eq. (\ref{eq:existence_condition_equivalent}) can be rewritten as
\begin{align}
 \max_{a_{-j}} \sum_{a_j}  \underline{p}_j(a_j) B\left( a_j, a_{-j} \right) &\leq 0 \nonumber \\
 \min_{a_{-j}} \sum_{a_j}  \overline{p}_j(a_j) B\left( a_j, a_{-j} \right) &\geq 0.
\end{align}
If such $\underline{p}_j$ and $\overline{p}_j$ exist, the inequalities (\ref{eq:existence_condition_minmax}) are satisfied, because
\begin{align}
 \min_{p_j} \max_{a_{-j}} \sum_{a_j}  p_j(a_j) B\left( a_j, a_{-j} \right) &\leq \max_{a_{-j}} \sum_{a_j}  \underline{p}_j(a_j) B\left( a_j, a_{-j} \right) \leq 0 \nonumber \\
 \max_{p_j} \min_{a_{-j}} \sum_{a_j}  p_j(a_j) B\left( a_j, a_{-j} \right) &\geq \min_{a_{-j}} \sum_{a_j}  \overline{p}_j(a_j) B\left( a_j, a_{-j} \right) \geq 0.
\end{align}
If Eq. (\ref{eq:existence_condition_minmax}) holds, we can obtain $\underline{p}_j$ and $\overline{p}_j$ by defining
\begin{align}
 \underline{p}_j &= \arg \min_{p_j} \max_{a_{-j}} \sum_{a_j}  p_j(a_j) B\left( a_j, a_{-j} \right) \nonumber \\
 \overline{p}_j &= \arg \max_{p_j} \min_{a_{-j}} \sum_{a_j}  p_j(a_j) B\left( a_j, a_{-j} \right).
\end{align}
\end{proof}
It should be noted that, as we will see in subsection \ref{subsec:PD}, $\underline{p}_j$ and $\overline{p}_j$ are generally not unique.

We now state our main theorem.
\begin{theorem}
\label{thm:broaden}
Mixed extension broadens the range of possible ZD strategies compared to the original repeated games.
\end{theorem}
We prove this theorem by explicitly constructing an example in subsection \ref{subsec:two-three}, which does not satisfy the condition in Theorem \ref{thm:pureZDS} but contains a ZD strategy in a continuously-relaxed action space.

\subsection{Prisoner's dilemma game}
\label{subsec:PD}
As an example, we consider the prisoner's dilemma game \cite{RCO1965}, where $\mathcal{N}=\{ 1, 2 \}$, $A_j=\{ C, D \}$ $(j=1, 2)$, and the payoffs are given as Table \ref{table:PD}.
We assume that $T>R>P>S$ and $2R>T+S$.
\begin{table}[tb]
  \centering
  \caption{Payoffs of the prisoner's dilemma game.}
  \begin{tabular}{|c|cc|} \hline
    & $C$ & $D$ \\ \hline
   $C$ & $R, R$ & $S, T$ \\
   $D$ & $T, S$ & $P, P$ \\ \hline
  \end{tabular}
  \label{table:PD}
\end{table}

Here we consider the existence of an equalizer-type strategy of player $1$, which unilaterally sets the payoff of the opponent \cite{BNS1997,PreDys2012}.
For such ZD strategies, we need to set
\begin{align}
 B\left( \bm{a} \right) &= s_2\left( \bm{a} \right) - r
 \label{eq:B_equalizer_PD}
\end{align}
with $S\leq r \leq T$.
When we write $p_j=\left( p_j(C), p_j(D) \right)^\mathsf{T}$ $(j=1,2)$, $\tilde{B}$ is given by
\begin{align}
 \tilde{B} \left( \bm{p} \right) &= p_1^\mathsf{T} \left(
\begin{array}{cc}
R-r & T-r \\
S-r & P-r
\end{array}
\right) p_2 \nonumber \\
&= \left(
\begin{array}{cc}
(R-r)p_1(C) + (S-r)\left( 1- p_1(C) \right) & (T-r)p_1(C) + (P-r)\left( 1- p_1(C) \right)
\end{array}
\right) p_2.
\end{align}
When $r>R$, we find
\begin{align}
 (R-r)p_1(C) + (S-r)\left( 1- p_1(C) \right) &< 0
\end{align}
for arbitrary $p_1$.
Therefore, according to Lemma \ref{lem:existence_equivalent}, $\overline{p}_1$ cannot exist, and we cannot construct a corresponding two-point ZD strategy.
Similarly, when $r<P$, we find 
\begin{align}
  (T-r)p_1(C) + (P-r)\left( 1- p_1(C) \right) &> 0
\end{align}
for arbitrary $p_1$.
Therefore, according to Lemma \ref{lem:existence_equivalent}, $\underline{p}_1$ cannot exist, and we cannot construct a corresponding two-point ZD strategy.
When $P\leq r \leq R$, we obtain
\begin{align}
 \tilde{B} \left( \bm{p} \right) &= \left(
\begin{array}{cc}
\left| R-r \right| p_1(C) - \left| r-S \right| \left( 1- p_1(C) \right) & \left| T-r \right| p_1(C) - \left| r-P \right| \left( 1- p_1(C) \right)
\end{array}
\right) p_2 \nonumber \\
&= \left(
\begin{array}{cc}
 \left| R-S \right| p_1(C) - \left| r-S \right| & \left| T-P \right| p_1(C) - \left| r-P \right|
\end{array}
\right) p_2.
\end{align}
Therefore, any $p_1$ such that
\begin{align}
 p_1(C) &\leq \min \left\{ \frac{\left| r-S \right|}{\left| R-S \right|}, \frac{\left| r-P \right|}{\left| T-P \right|} \right\}
\end{align}
can be $\underline{p}_1$.
Similarly, any $p_1$ such that
\begin{align}
 p_1(C) &\geq \max \left\{ \frac{\left| r-S \right|}{\left| R-S \right|}, \frac{\left| r-P \right|}{\left| T-P \right|} \right\}
\end{align}
can be $\overline{p}_1$.
By using such $\underline{p}_1$ and $\overline{p}_1$, we can construct a two-point ZD strategy of player $1$ unilaterally enforcing
\begin{align}
 \mathcal{U}_2 &= r.
\end{align}

In Figure \ref{fig:linear_PD}, we display a relation between $\mathcal{U}_1$ and $\mathcal{U}_2$ when player $1$ uses the two-point ZD strategy (\ref{eq:transition_two-point}) and player $2$ uses randomly-chosen two-point memory-one strategies.
\begin{figure}[tbp]
\begin{center}
\includegraphics[clip, width=8.0cm]{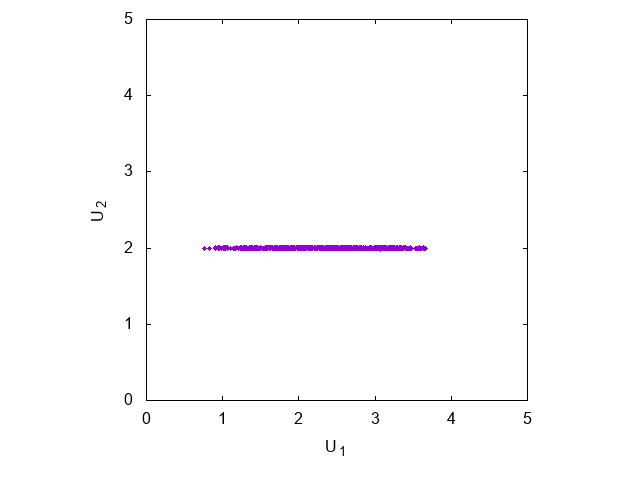}
\end{center}
\caption{A linear relation between $\mathcal{U}_1$ and $\mathcal{U}_2$ when player $1$ uses the two-point ZD strategy (\ref{eq:transition_two-point}) and player $2$ uses $10^3$ randomly-chosen two-point memory-one strategies. Parameters are set to $(R, S, T, P)=(3, 0, 5, 1)$, $r=2$, $W=3$, $\overline{p}_1(C)=2/3$, and $\underline{p}_1(C)=1/4$. Each $\mathcal{U}_j$ is calculated by time average over $10^5$ time steps.}
\label{fig:linear_PD}
\end{figure}
We find that a linear relation $\mathcal{U}_2 = r$ is indeed enforced.

It should be noted that Eq. (\ref{eq:B_equalizer_PD}) satisfies
\begin{align}
 B\left( D, a_2 \right) &= s_2\left( D, a_2 \right) - r \leq 0 \quad (\forall a_2 \in A_2) \nonumber \\
 B\left( C, a_2 \right) &= s_2\left( C, a_2 \right) - r \geq 0 \quad (\forall a_2 \in A_2)
 \label{eq:B_equalizer_PD_two}
\end{align}
for $P\leq r \leq R$.
Therefore, the condition in Theorem \ref{thm:pureZDS} is also satisfied, and we can construct the equalizer strategy by using pure strategies $\underline{p}_1=(0,1)^\mathsf{T}$ and $\overline{p}_1=(1,0)^\mathsf{T}$.
As noted above, a relation between Eq. (\ref{eq:B_equalizer_PD_two}) and the existence of the original equalizer strategy is a well-known fact \cite{Ued2022b}.

\subsection{Public goods game}
\label{subsec:PGG}
As another example, we consider the public goods game \cite{HWTN2014,PHRT2015}, where $\mathcal{N}=\{ 1, \cdots, N \}$ $(N\geq 2)$, $A_j=\{ C, D \}$ $(j\in \mathcal{N})$, and the payoffs are given by
\begin{align}
 s_j (\bm{a}) &= \frac{rc}{N} \sum_{l\neq j} \delta_{a_l, C} + c\left( \frac{r}{N} - 1 \right) \delta_{a_j, C}
\end{align}
with $c>0$ and $1<r<N$.
The public goods game can be regarded as an $N$-player version of the prisoner's dilemma game.

Here we again consider the existence of an equalizer-type ZD strategy of player $j$.
For such ZD strategies, we need to set
\begin{align}
 B\left( \bm{a} \right) &= \sum_{k\neq j} s_k\left( \bm{a} \right) - \mu \nonumber \\
 &= \left( \frac{N-1}{N}r - 1 \right) c \sum_{k\neq j} \delta_{a_k, C} +  \frac{N-1}{N}rc \delta_{a_j, C} - \mu.
 \label{eq:B_equalizer_PGG}
\end{align}
Below we write $p_j=\left( p_j(C), p_j(D) \right)^\mathsf{T}$.
In order to apply Lemma \ref{lem:existence_equivalent}, we consider two cases $1<r \leq N/(N-1)$ and $N/(N-1)<r<N$ separately.
\begin{enumerate}[(i)]
\item $1<r \leq N/(N-1)$\\
For this case, we can calculate as
\begin{align}
 \max_{a_{-j}} \sum_{a_j} p_j(a_j) B\left( \bm{a} \right) &= \frac{N-1}{N}rc p_j(C)  - \mu \nonumber \\
 \min_{a_{-j}} \sum_{a_j}  p_j(a_j) B\left( \bm{a} \right) &= \left( \frac{N-1}{N}r - 1 \right) c (N-1) + \frac{N-1}{N}rc p_j(C)  - \mu.
\end{align}
Therefore, we obtain
\begin{align}
 \min_{p_j} \max_{a_{-j}} \sum_{a_j} p_j(a_j) B\left( \bm{a} \right) &= -\mu \nonumber \\
 \max_{p_j} \min_{a_{-j}} \sum_{a_j}  p_j(a_j) B\left( \bm{a} \right) &= \left( \frac{N-1}{N}r - 1 \right) c (N-1) + \frac{N-1}{N}rc - \mu.
\end{align}
According to Lemma \ref{lem:existence_equivalent}, two mixed strategies $\underline{p}_j$ and $\overline{p}_j$ exist if and only if the inequalities
\begin{align}
 -\mu &\leq 0 \nonumber \\
 \left( \frac{N-1}{N}r - 1 \right) c (N-1) + \frac{N-1}{N}rc - \mu &\geq 0
\end{align}
are satisfied, which are simplified as
\begin{align}
 0 \leq \mu \leq (N-1)(r-1)c.
 \label{eq:mu_PGG_low}
\end{align}
\item $N/(N-1)<r<N$\\
For this case, we can calculate as
\begin{align}
 \max_{a_{-j}} \sum_{a_j} p_j(a_j) B\left( \bm{a} \right) &= \left( \frac{N-1}{N}r - 1 \right) c (N-1) + \frac{N-1}{N}rc p_j(C) - \mu \nonumber \\
 \min_{a_{-j}} \sum_{a_j}  p_j(a_j) B\left( \bm{a} \right) &= \frac{N-1}{N}rc p_j(C) -\mu.
\end{align}
Therefore, we obtain
\begin{align}
 \min_{p_j} \max_{a_{-j}} \sum_{a_j} p_j(a_j) B\left( \bm{a} \right) &= \left( \frac{N-1}{N}r - 1 \right) c (N-1) -\mu \nonumber \\
 \max_{p_j} \min_{a_{-j}} \sum_{a_j}  p_j(a_j) B\left( \bm{a} \right) &= \frac{N-1}{N}rc - \mu.
\end{align}
According to Lemma \ref{lem:existence_equivalent}, two mixed strategies $\underline{p}_j$ and $\overline{p}_j$ exist if and only if the inequalities
\begin{align}
 \left( \frac{N-1}{N}r - 1 \right) c (N-1) -\mu &\leq 0 \nonumber \\
 \frac{N-1}{N}rc - \mu &\geq 0,
\end{align}
are satisfied, which are simplified as
\begin{align}
 \left( \frac{N-1}{N}r - 1 \right) c (N-1) \leq \mu \leq \frac{N-1}{N}rc.
\end{align}
It should be noted that such range of $\mu$ exists only for
\begin{align}
 \frac{N}{N-1} < r \leq \frac{N}{N-2}.
\end{align}
\end{enumerate}

These results on the existence of equalizer strategies coincide with those for pure strategy case \cite{PHRT2015}.
(We remark that our definition of the payoffs is slightly different from that in Ref. \cite{PHRT2015}.)
This is because we can check that two actions $C$ and $D$ satisfy
\begin{align}
 B\left( D, a_{-j} \right) &\leq 0 \quad (\forall a_{-j} \in A_{-j}) \nonumber \\
 B\left( C, a_{-j} \right) &\geq 0 \quad (\forall a_{-j} \in A_{-j})
\end{align}
for the above parameter regions.
Therefore, similarly to the prisoner's dilemma game, the condition in Theorem \ref{thm:pureZDS} is satisfied, and we can construct the equalizer strategy by using two pure strategies.

Moreover, as in the prisoner's dilemma game, we can use mixed strategies as $\underline{p}_j$ and $\overline{p}_j$ in order to construct ZD strategies.
For example, when $1<r \leq N/(N-1)$ and Eq. (\ref{eq:mu_PGG_low}) hold, any $p_j$ satisfying
\begin{align}
 \max_{a_{-j}} \sum_{a_j} p_j(a_j) B\left( \bm{a} \right) &= \frac{N-1}{N}rc  p_j(C) -\mu \leq 0
\end{align}
can be $\underline{p}_j$, and any $p_j$ satisfying
\begin{align}
 \min_{a_{-j}} \sum_{a_j}  p_j(a_j) B\left( \bm{a} \right) &= \left( \frac{N-1}{N}r - 1 \right) c (N-1) + \frac{N-1}{N}rc p_j(C) - \mu \geq 0
\end{align}
can be $\overline{p}_j$.

\subsection{A two-player three-action symmetric game}
\label{subsec:two-three}
As another example, we consider a two-player three-action symmetric game, where $\mathcal{N}=\{ 1, 2 \}$, $A_j=\{ 1, 2, 3 \}$ $(j=1, 2)$, and the payoffs are given as Table \ref{table:twothree}.
\begin{table}[tb]
  \centering
  \caption{Payoffs of a two-player three-action symmetric game.}
  \begin{tabular}{|c|ccc|} \hline
    & $1$ & $2$ & $3$ \\ \hline
   $1$ & $-3, -3$ & $1, 2$ & $-1, -\frac{7}{3}$ \\
   $2$ & $2, 1$ & $-1, -1$ & $2, 1$ \\
   $3$ & $-\frac{7}{3}, -1$ & $1, 2$ & $-\frac{3}{2}, -\frac{3}{2}$ \\ \hline
  \end{tabular}
  \label{table:twothree}
\end{table}

When we set
\begin{align}
 B\left( \bm{a} \right) &= s_2\left( \bm{a} \right)
\end{align}
and write $p_j=\left( p_j(1), p_j(2), p_j(3) \right)^\mathsf{T}$ $(j=1,2)$, $\tilde{B}$ is given by
\begin{align}
 \tilde{B} \left( \bm{p} \right) &= p_1^\mathsf{T} \left(
\begin{array}{ccc}
-3 & 2 & -\frac{7}{3}\\
1 & -1 & 1 \\
-1 & 2 & -\frac{3}{2} \\
\end{array}
\right) p_2.
\end{align}
Then, we find
\begin{align}
 \tilde{B} \left( \left( \frac{1}{3}, \frac{2}{3}, 0 \right)^\mathsf{T}, p_2 \right) &= \left(
\begin{array}{ccc}
-\frac{1}{3} & 0 & -\frac{1}{9}
\end{array}
\right) p_2 \leq 0 \quad (\forall p_2 \in \Delta(A_2))
\end{align}
and
\begin{align}
 \tilde{B} \left( \left( 0, \frac{2}{3}, \frac{1}{3} \right)^\mathsf{T}, p_2 \right) &= \left(
\begin{array}{ccc}
\frac{1}{3} & 0 & \frac{1}{6}
\end{array}
\right) p_2 \geq 0 \quad (\forall p_2 \in \Delta(A_2)).
\end{align}
Therefore, we can regard $\underline{p}_1=(1/3, 2/3, 0)^\mathsf{T}$ and $\overline{p}_1=(0, 2/3, 1/3)^\mathsf{T}$ in Lemma \ref{lem:existence_equivalent}, and can construct a two-point ZD strategy of player $1$, which unilaterally enforces
\begin{align}
 \mathcal{U}_2 &= 0,
\end{align}
by using Proposition \ref{prop:two-point}.

In Figure \ref{fig:linear_sym}, we display a relation between $\mathcal{U}_1$ and $\mathcal{U}_2$ when player $1$ uses the two-point ZD strategy (\ref{eq:transition_two-point}) and player $2$ uses randomly-chosen two-point memory-one strategies.
\begin{figure}[tbp]
\begin{center}
\includegraphics[clip, width=8.0cm]{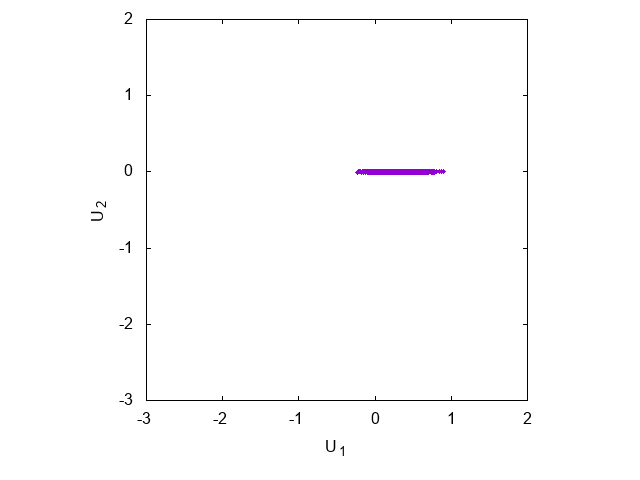}
\end{center}
\caption{A linear relation between $\mathcal{U}_1$ and $\mathcal{U}_2$ when player $1$ uses the two-point ZD strategy (\ref{eq:transition_two-point}) and player $2$ uses $10^4$ randomly-chosen two-point memory-one strategies. Parameter is set to $W=3$. Each $\mathcal{U}_j$ is calculated by time average over $10^5$ time steps.}
\label{fig:linear_sym}
\end{figure}
We find that a linear relation $\mathcal{U}_2 = 0$ is indeed enforced.

We remark that none of the actions of player $1$ becomes $\underline{a}_1$ or $\overline{a}_1$ in Eq. (\ref{eq:existence_condition_pure}) in Theorem \ref{thm:pureZDS}.
Therefore, differently from the prisoner's dilemma case, the corresponding pure strategy game does not contain a ZD strategy controlling $s_2$, and we conclude that mixed extension broadens the range of possible ZD strategies compared to the original repeated games (Theorem \ref{thm:broaden}).

This fact can also be understood by displaying $\underline{p}_1$ and $\overline{p}_1$ in $\Delta(A_1)$ as in Figure \ref{fig:simplex_twothree}.
\begin{figure}[tbp]
\begin{center}
\includegraphics[clip, width=10.0cm]{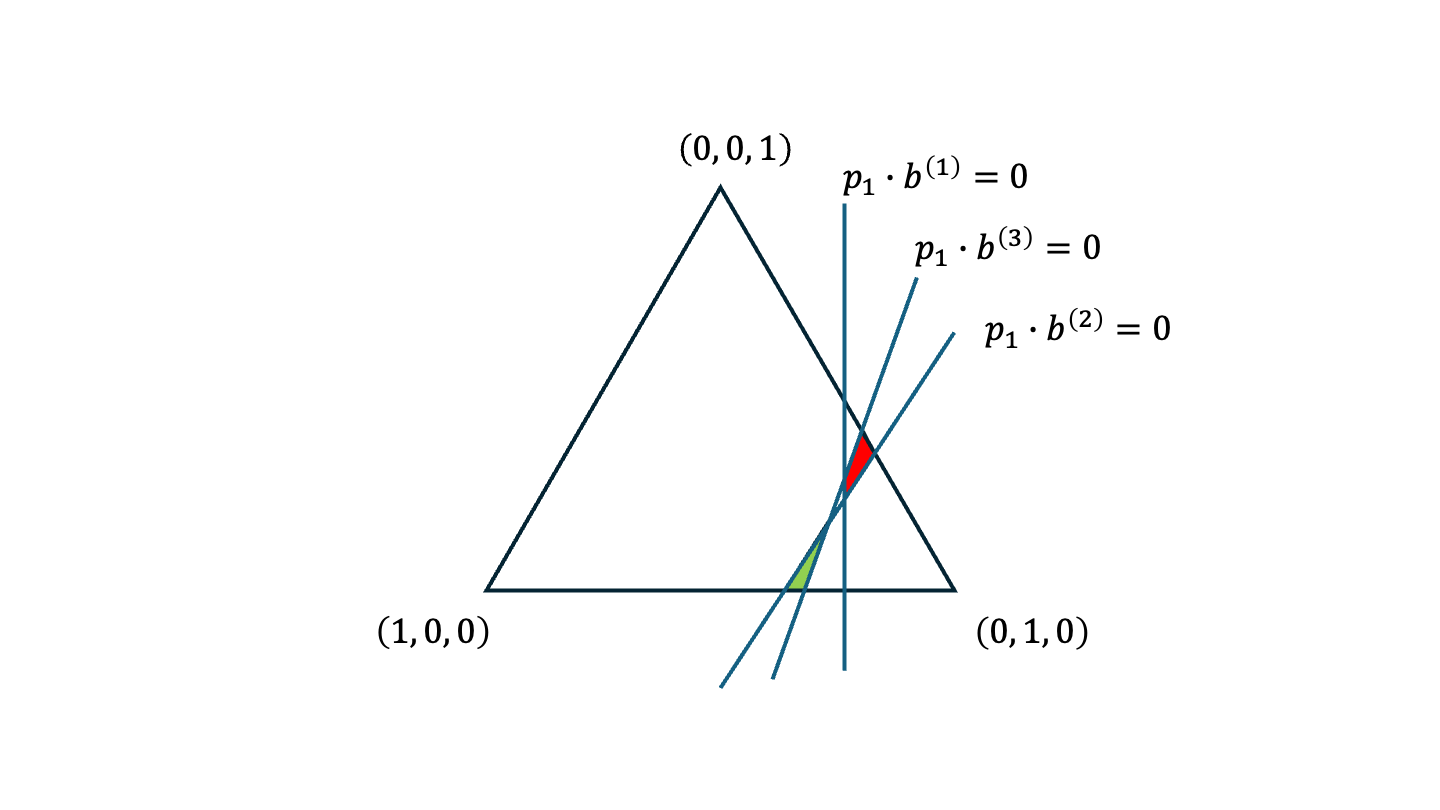}
\end{center}
\caption{Feasible regions of $\underline{p}_1$ and $\overline{p}_1$ in a probability simplex $\Delta(A_1)$. We define $b^{(a_2)}(a_1):= B(a_1, a_2)$ for simplicity. Any $p_1$ in the green region can be used as $\underline{p}_1$, and any $p_1$ in the red region can be used as $\overline{p}_1$.}
\label{fig:simplex_twothree}
\end{figure}
In this figure, we have defined $b^{(a_2)}(a_1):= B(a_1, a_2)$ for simplicity.
According to Lemma \ref{lem:existence_equivalent}, $p_1$ satisfying $p_1\cdot b^{(a_2)} \leq 0$ $(\forall a_2)$, if exists, can be used as $\underline{p}_1$, and $p_1$ satisfying $p_1\cdot b^{(a_2)} \geq 0$ $(\forall a_2)$, if exists, can be used as $\overline{p}_1$.
Such two $p_1$ exist for this case, and $p_1$ in the green and red region can be $\underline{p}_1$ and $\overline{p}_1$, respectively.
It should be noted that each vertex of the equilateral triangle corresponds to a pure strategy, and is not contained in both regions.

\section{One-point ZD strategies}
\label{sec:one-point}
In this section, we consider a special situation where $\underline{p}_j=\overline{p}_j=:p_j^{(0)}$ in the condition (\ref{eq:existence_condition}), which implies
\begin{align}
 \tilde{B}\left(p_j^{(0)}, p_{-j} \right) &= 0 \quad (\forall p_{-j} \in \tilde{\Delta}(A_{-j})).
 \label{eq:B_one-point}
\end{align}
For such case, a two-point ZD strategy in Proposition \ref{prop:two-point} becomes
\begin{align}
 T_j\left( p_j | \bm{p}^\prime \right) &= \delta\left( p_j-p_j^\prime \right) \quad (\forall p_j^\prime \in \Delta\left(A_j\right)^\prime, \forall p_{-j}^\prime \in \tilde{\Delta}(A_{-j}))
 \label{eq:transition_one-point}
\end{align}
defined on $\Delta\left(A_j\right)^\prime=\left\{ p_j^{(0)} \right\}$.
This strategy satisfies
\begin{align}
 \int \psi(p_j) T_j\left( p_j | \bm{p}^\prime \right) dp_j - \psi\left( p^\prime_j \right) &= 0 = \tilde{B}\left( \bm{p}^\prime \right) \quad (\forall p_j^\prime \in \Delta\left(A_j\right)^\prime, \forall p_{-j}^\prime \in \tilde{\Delta}(A_{-j}))
\end{align}
for arbitrary $\psi$.
Although such ``Repeat'' strategies \cite{Aki2016} have not been regarded as ZD strategies \cite{PreDys2012}, here we call such strategies as \emph{one-point ZD strategies} only if Eq. (\ref{eq:B_one-point}) holds, because they still unilaterally enforce Eq. (\ref{eq:linear_AS}).
This fact is summarized as a corollary of Proposition \ref{prop:two-point}.
\begin{corollary}
\label{cor:one-point}
If there exists a mixed strategy $p_j^{(0)} \in \Delta(A_j)$ satisfying Eq. (\ref{eq:B_one-point}), then the memory-one strategy (\ref{eq:transition_one-point}) of player $j$ defined on $\Delta\left(A_j\right)^\prime=\left\{ p_j^{(0)} \right\}$ is a ZD strategy controlling $\tilde{B}$.
\end{corollary}

Below we provide several examples of one-point ZD strategies.

\subsection{Matching pennies}
First, we consider the matching pennies game \cite{FudTir1991}, where $\mathcal{N}=\{ 1, 2 \}$, $A_j=\{ 1, 2 \}$ $(j=1, 2)$, and the payoffs are given as Table \ref{table:MP}.
\begin{table}[tb]
  \centering
  \caption{Payoffs of the matching pennies game.}
  \begin{tabular}{|c|cc|} \hline
    & 1 & 2 \\ \hline
   1 & $1, -1$ & $-1, 1$ \\
   2 & $-1, 1$ & $1, -1$ \\ \hline
  \end{tabular}
  \label{table:MP}
\end{table}
Here we consider the existence of an equalizer-type strategy of player $1$:
\begin{align}
 B\left( \bm{a} \right) &= s_2\left( \bm{a} \right) - r
 \label{eq:B_equalizer_MP}
\end{align}
with $-1\leq r \leq 1$.
When we write $p_j=\left( p_j(1), p_j(2) \right)^\mathsf{T}$ $(j=1,2)$, $\tilde{B}$ is given by
\begin{align}
 \tilde{B} \left( \bm{p} \right) &= p_1^\mathsf{T} \left(
\begin{array}{cc}
-1-r & 1-r \\
1-r & -1-r
\end{array}
\right) p_2 \nonumber \\
&= \left(
\begin{array}{cc}
1-2p_1(1) - r & -1+2p_1(1) - r
\end{array}
\right) p_2.
\end{align}
According to Lemma \ref{lem:existence_equivalent}, for the existence of two-point ZD strategies, two mixed strategies $\underline{p}_1, \overline{p}_1$ satisfying
\begin{align}
 1-2\underline{p}_1(1) - r &\leq 0 \nonumber \\
 -1+2\underline{p}_1(1) - r &\leq 0 \nonumber \\
 1-2\overline{p}_1(1) - r &\geq 0 \nonumber \\
 -1+2\overline{p}_1(1) - r &\geq 0 
\end{align}
are necessary.
By summing the first and second inequalities, and the third and fourth inequalities, we obtain
\begin{align}
 r &\geq 0 \nonumber \\
 r &\leq 0,
\end{align}
respectively.
Therefore, $r=0$ must hold.
Furthermore, these inequalities are satisfied only for
\begin{align}
  \underline{p}_1(1) &= \frac{1}{2} \nonumber \\
  \overline{p}_1(1) &= \frac{1}{2}.
\end{align}
Thus, $\underline{p}_1=\overline{p}_1=(1/2, 1/2)^\mathsf{T}$, and we obtain
\begin{align}
 \tilde{B} \left( \left( \frac{1}{2}, \frac{1}{2} \right)^\mathsf{T}, p_2 \right) &= \left(
\begin{array}{cc}
0 & 0
\end{array}
\right) p_2 \nonumber \\
&= 0 \quad (\forall p_2 \in \Delta(A_2)).
\end{align}
Therefore, by using this mixed strategy, we can construct a one-point ZD strategy which unilaterally enforces
\begin{align}
 \mathcal{U}_2 &= 0.
\end{align}

We remark that this result is well-known as a property of the mixed-strategy Nash equilibrium, where every action in the support of any player's equilibrium mixed strategy yields that player the same payoff \cite{OsbRub1994}.
In the matching pennies game, a mixed-strategy profile $(p_1, p_2) = \left( (1/2, 1/2)^\mathsf{T}, (1/2, 1/2)^\mathsf{T} \right)$ is the Nash equilibrium, and once player $1$ uses a mixed strategy $p_1=(1/2, 1/2)^\mathsf{T}$, two actions of player $2$ yield the same payoff.

We also remark that none of the actions of player $1$ becomes $\underline{a}_1$ or $\overline{a}_1$ in Eq. (\ref{eq:existence_condition_pure}) with $r=0$ in Theorem \ref{thm:pureZDS}.
Therefore, this is another example of Theorem \ref{thm:broaden}.

\subsection{Battle of the sexes}
Next, we consider the battle of the sexes game \cite{OsbRub1994}, where $\mathcal{N}=\{ 1, 2 \}$, $A_j=\{ 1, 2 \}$ $(j=1, 2)$, and the payoffs are given as Table \ref{table:BS}.
\begin{table}[tb]
  \centering
  \caption{Payoffs of the battle of the sexes game.}
  \begin{tabular}{|c|cc|} \hline
    & 1 & 2 \\ \hline
   1 & $2, 1$ & $0, 0$ \\
   2 & $0, 0$ & $1, 2$ \\ \hline
  \end{tabular}
  \label{table:BS}
\end{table}
Again, we consider the existence of an equalizer-type strategy of player $1$:
\begin{align}
 B\left( \bm{a} \right) &= s_2\left( \bm{a} \right) - r
 \label{eq:B_equalizer_BS}
\end{align}
with $0\leq r \leq 2$.
When we write $p_j=\left( p_j(1), p_j(2) \right)^\mathsf{T}$ $(j=1,2)$, $\tilde{B}$ is given by
\begin{align}
 \tilde{B} \left( \bm{p} \right) &= p_1^\mathsf{T} \left(
\begin{array}{cc}
1-r & -r \\
-r & 2-r
\end{array}
\right) p_2 \nonumber \\
&= \left(
\begin{array}{cc}
p_1(1) - r & 2-2p_1(1) - r
\end{array}
\right) p_2.
\end{align}
According to Lemma \ref{lem:existence_equivalent}, for the existence of two-point ZD strategies, two mixed strategies $\underline{p}_1, \overline{p}_1$ satisfying
\begin{align}
 \underline{p}_1(1) - r &\leq 0 \nonumber \\
 2-2\underline{p}_1(1) - r &\leq 0 \nonumber \\
 \overline{p}_1(1) - r &\geq 0 \nonumber \\
 2-2\overline{p}_1(1) - r &\geq 0 
\end{align}
are necessary.
From the first and second inequalities, and the third and fourth inequalities, we obtain
\begin{align}
 2-3r &\leq 0 \nonumber \\
 2-3r &\geq 0,
\end{align}
respectively.
Therefore, $r=2/3$ must hold.
Furthermore, these inequalities are satisfied only for
\begin{align}
  \underline{p}_1(1) &= \frac{2}{3} \nonumber \\
  \overline{p}_1(1) &= \frac{2}{3}.
\end{align}
Thus, $\underline{p}_1=\overline{p}_1=(2/3, 1/3)^\mathsf{T}$, and we obtain
\begin{align}
 \tilde{B} \left( \left( \frac{2}{3}, \frac{1}{3} \right)^\mathsf{T}, p_2 \right) &= \left(
\begin{array}{cc}
0 & 0
\end{array}
\right) p_2 \nonumber \\
&= 0 \quad (\forall p_2 \in \Delta(A_2)).
\end{align}
Therefore, by using this mixed strategy, we can construct a one-point ZD strategy which unilaterally enforces
\begin{align}
 \mathcal{U}_2 - \frac{2}{3} &= 0.
\end{align}
We remark that this result is also well-known as a property of the mixed-strategy Nash equilibrium $(p_1, p_2) = \left( (2/3, 1/3)^\mathsf{T}, (1/3, 2/3)^\mathsf{T} \right)$.

Again, none of the actions of player $1$ becomes $\underline{a}_1$ or $\overline{a}_1$ in Eq. (\ref{eq:existence_condition_pure}) with $r=2/3$ in Theorem \ref{thm:pureZDS}.
Therefore, this is another example of Theorem \ref{thm:broaden}.

\subsection{Rock-paper-scissors game}
Here, we consider the rock-paper-scissors game, where $\mathcal{N}=\{ 1, 2 \}$, $A_j=\{ R, P, S \}$ $(j=1, 2)$, and the payoffs are given as Table \ref{table:RPS}.
\begin{table}[tb]
  \centering
  \caption{Payoffs of the rock-paper-scissors game.}
  \begin{tabular}{|c|ccc|} \hline
    & $R$ & $P$ & $S$ \\ \hline
   $R$ & $0, 0$ & $-1, 1$ & $1, -1$ \\
   $P$ & $1, -1$ & $0, 0$ & $-1, 1$ \\
   $S$ & $-1, 1$ & $1, -1$ & $0, 0$ \\ \hline
  \end{tabular}
  \label{table:RPS}
\end{table}
We consider the existence of an equalizer-type strategy of player $1$:
\begin{align}
 B\left( \bm{a} \right) &= s_2\left( \bm{a} \right) - r
 \label{eq:B_equalizer_RPS}
\end{align}
with $-1\leq r \leq 1$.
When we write $p_j=\left( p_j(R), p_j(P), p_j(S) \right)^\mathsf{T}$ $(j=1,2)$, $\tilde{B}$ is given by
\begin{align}
 \tilde{B} \left( \bm{p} \right) &= p_1^\mathsf{T} \left(
\begin{array}{ccc}
-r & 1-r & -1-r \\
-1-r & -r & 1-r \\
1-r & -1-r & -r \\
\end{array}
\right) p_2 \nonumber \\
&= \left(
\begin{array}{ccc}
-p_1(P)+p_1(S) - r & p_1(R)-p_1(S) - r & -p_1(R)+p_1(P) - r
\end{array}
\right) p_2.
\end{align}
According to Lemma \ref{lem:existence_equivalent}, for the existence of two-point ZD strategies, two mixed strategies $\underline{p}_1, \overline{p}_1$ satisfying
\begin{align}
 -\underline{p}_1(P)+\underline{p}_1(S) - r &\leq 0 \nonumber \\
\underline{p}_1(R)-\underline{p}_1(S) - r &\leq 0 \nonumber \\
-\underline{p}_1(R)+\underline{p}_1(P) - r &\leq 0 \nonumber \\
 -\overline{p}_1(P)+\overline{p}_1(S) - r &\geq 0 \nonumber \\
\overline{p}_1(R)-\overline{p}_1(S) - r &\geq 0 \nonumber \\
-\overline{p}_1(R)+\overline{p}_1(P) - r &\geq 0 
\end{align}
are necessary.
By summing the inequalities, we obtain
\begin{align}
 -r &\leq 0 \nonumber \\
 -r &\geq 0.
\end{align}
Therefore, $r=0$ must hold.
Furthermore, these inequalities are satisfied only for
\begin{align}
  \underline{p}_1(R) &= \underline{p}_1(P) = \underline{p}_1(S) \nonumber \\
  \overline{p}_1(R) &= \overline{p}_1(P) = \overline{p}_1(S).
\end{align}
Thus, $\underline{p}_1=\overline{p}_1=(1/3, 1/3, 1/3)^\mathsf{T}$, and we obtain
\begin{align}
 \tilde{B} \left( \left( \frac{1}{3}, \frac{1}{3}, \frac{1}{3} \right)^\mathsf{T}, p_2 \right) &= \left(
\begin{array}{ccc}
0 & 0 & 0
\end{array}
\right) p_2 \nonumber \\
&= 0 \quad (\forall p_2 \in \Delta(A_2)).
\end{align}
Therefore, by using this mixed strategy, we can construct a one-point ZD strategy which unilaterally enforces
\begin{align}
 \mathcal{U}_2 &= 0.
\end{align}
We remark that this result is also well-known as a property of the mixed-strategy Nash equilibrium $(p_1, p_2) = \left( (1/3, 1/3, 1/3)^\mathsf{T}, (1/3, 1/3, 1/3)^\mathsf{T} \right)$.

It should be noted that the original rock-paper-scissors game does not contain any ZD strategies \cite{UedTan2020}.
Therefore, this is another example that mixed extension broadens the range of possible ZD strategies compared to the original repeated games (Theorem \ref{thm:broaden}).

\section{Concluding remarks}
\label{sec:conclusion}
In this paper, we investigated the existence of two-point ZD strategies in repeated games where action sets in stage games are continuously relaxed.
Through an example in subsection \ref{subsec:two-three}, we found that the existence condition of two-point ZD strategies in continuously-relaxed action sets is weaker than that of two-point ZD strategies in the original finite action sets.
Furthermore, we introduced the concept of one-point ZD strategies, as a special case of two-point ZD strategies.
In our three examples, we found that a property of mixed-strategy Nash equilibria can be reinterpreted as a payoff-control property of one-point ZD strategies.

Before ending this paper, we provide three remarks.
The first remark is on our interpretation of mixed strategies.
In Section \ref{sec:model}, we assumed that we interpret a mixed strategy as a strategy realized by a population.
This interpretation of mixed strategies is necessary for our study, because a behavior strategy in Eq. (\ref{eq:behavior_strategy}) makes sense only when mixed strategies in previous rounds are explicitly observed.
If we interpret a mixed strategy as a probability distribution of actions taken by one player, each round must contain many games, which is an unnatural setup.
Therefore, we used the interpretation that a mixed strategy is a strategy realized by a population.
Unexpectedly, this interpretation is compatible with a standard assumption of evolutionary game theory \cite{HofSig1998}, where a mixed strategy describes a population.
Particularly, discrete-time coupled replicator equations in evolutionary game theory are a special case of our Markov chain with deterministic transition probability.
Therefore, if we seek for realistic situations to which we apply our zero-determinant strategies on mixed action spaces, evolutionary game theory is the most suitable candidate.
It should be noted that evolutionary performance of ZD strategies in structured populations has attracted much attention \cite{SzoPer2014,LLXH2015,DZWLW2019,SGCG2024,CGCF2024}.

The second remark is related to possible ZD strategies in continuously-relaxed action sets.
When action sets of all players are finite sets, the existence condition of two-point ZD strategies is equivalent to that of general ZD strategies \cite{Ued2025}.
At this stage, we do not know whether this is also true when action sets of some players are infinite sets.
One may expect that, when action sets of some players are infinite sets, ZD strategies can exist under weaker condition than one in Ref. \cite{Ued2022b}.
We will try constructing ZD strategies which do not satisfy the condition (\ref{eq:existence_condition}) in our continuously-relaxed action sets in future.

The third remark is on the relation between a property of mixed-strategy Nash equilibria and a property of one-point ZD strategies.
Although we provided three examples where equilibrium mixed strategies appear as $p_j^{(0)}$ in Section \ref{sec:one-point}, this is not a general result.
As noted in the section, in mixed-strategy Nash equilibria, only actions in the support of a player's equilibrium mixed strategy yield that player the same payoff.
Therefore, even in two-player games, the property of mixed-strategy Nash equilibria can be reinterpreted as a payoff-control property of one-point ZD strategies if and only if all actions are in the support of the equilibrium mixed strategy of the opponent.
Further investigation is needed to clarify the relation between them.

\section*{Acknowledgement}
This study was supported by JSPS KAKENHI Grant Number JP20K19884 and Toyota Riken Scholar Program.



\section*{References}

\bibliography{mixedZDS}

\end{document}